\documentclass[a4paper,12pt]{article}%
\usepackage{amsmath}
\usepackage{amsfonts}
\usepackage{amssymb}
\usepackage{graphicx}
\usepackage{authblk}
\usepackage{natbib}
\usepackage{float,enumerate}
\usepackage{color}
\usepackage{fancyhdr}
\usepackage{setspace}%
\providecommand{\U}[1]{\protect\rule{.1in}{.1in}}
\newtheorem{theorem}{Theorem}

\newtheorem{corollary}[theorem]{Corollary}

\newtheorem{lemma}[theorem]{Lemma}

\newtheorem{proposition}[theorem]{Proposition}

\newenvironment{proof}[1][Proof]{\noindent\textbf{#1.} }{\ \rule{0.5em}{0.5em}}
\begin{document}
\title{Default Contagion with Domino Effect \\
---A First Passage Time
Approach}

\author[$\star$]{Jir\^o Akahori 
}
\author[$\dagger$]{Hai-Ha Pham}

\affil[$\star$]{Department of Mathematical Sciences, Ritsumeikan University, Japan}
\affil[$\dagger$]{Department of Economical Mathematics, Banking University of Hochiminh city, Vietnam}

\date{}
\maketitle
\section{Introduction}

The systematic risk and contagion play a crucial role in the financial crisis.
Many recent researches have put attention on this object since the Asian banking crisis of the late 90s, and the more recent banking crisis of 2007-2008. Most of them used directed graphs (network) to model interdependencies of system finance. For example, in \cite{elliott2014financial},
values of organizations depend on each other - e.g., through cross-holdings of shares, debt, or other liabilities. 
By tracking how asset values and failure costs propagate through the network of interdependencies as domino effect, the authors show how the probability of cascades and their extent depend on two key aspects of cross-holdings:
integration and diversification. \cite{rogers2013failure} is interested in the role of linkage in interbank and provide condition when rescue consortia exists. 
In order to study the importance of institutions, one can use the Contagion Index \cite{cont2010network},
CDS spreads or equity volatility \cite{acharya2017measuring}, \cite{huang2009framework}... 
On the contrary, the global level of systemic risk in the entire network can be considered by many other authors. References on the global level is referred to \cite{cont2010network} for a short survey.
In another approach, \cite{acemoglu2015systemic} is interested in the stability and resilience of the financial system under effecting of negative shocks.
\cite{fouque2013systemic} focuses on the number of components reaching a “default” level in a given time. The authors used the mean-field limit and a large deviation estimate for a simple linear model of lending and borrowing banks to illustrate systematic risk. 
The mean-field limit is also used 
in \cite{chong2015partial} to prove a law of large numbers. 
Both articles investigate system of interacting stochastic differential equations.  

The present paper introduces a structural framework to model dependent defaults, with a particular interest in their contagion. The idea is based on the modeling of dependent stochastic intensities in \cite{jarrow2001counterparty}. 
By ``structural", we mean, as usual in the context, modeling bankruptcy or {\em default} of a contingent claim by an event that a {\em firm value process} reaches a level.
The level can be either endogenous as the seminal paper 
\cite{leland1996optimal}, but at this stage it is given exogenously. Since we are concerned with mutual dependence of the default, we consider a vector valued process, each component of which is the {\em value} of a firm. 
In our model, to describe the contagion, the default level of each firm is assumed to be affected every time a default of another firm occurs. Such kind of model has also been studied in \cite{chong2016contagion}. 
In their paper \cite{chong2016contagion}, the authors use the Bayesian network methodology to characterize the joint default distribution of the financial system at a given maturity.

In our model, a default of a firm 
brings about a prescribed constant jump 
to the default level of other firms. 
One default can therefore cause other 
defaults, but each of the second order default 
may trigger third order ones, and so on.
This gives a structural framework to Bayesian network type dependence of joint default probability. 

Another difference is that \cite{chong2016contagion} only considers the firms value at maturity time. The default or survive of a company is determined by its equity. It is different from our approach. We are not only interested in the number of default at given time but also in default time and number of default at default time. They depend on the state of firms value which hits some special zone, called contagion region.

\section{Model}\label{SectionModel}
The model used in the present paper is similar to the one in \cite{chong2016contagion}. 
Let $ X^i_t $ denote 
the firm value process of the $ i $-th company, for $ i=1, 2, \cdots, n $ with
$ n \geq 2 $. 
Define ``default time" by
\begin{equation*}
\tau^i := \inf \{ s \geq 0: X_s^i \leq K^i \}, 
\end{equation*}
where $ k^i \in \mathbf{R} $ is a exogenously  given 
default level for the $ i $-th company. 
We assume that $ X \equiv (X^1, \cdots, X^n) $ solves the following equation; 
\begin{equation}\label{general}
\begin{split}
X^i_t &= x^i -
\sum_{j \ne i} C_{i,j} 1_{\{ \tau^j < t \wedge \tau^i \}}
+ \int_0^{t \wedge \tau^i} 
( \sigma_i (X^i_s) dW^i
+ \mu^i (X_s^i) \, dt) \\ 
&+\sum_{j \ne i} \int_{t \wedge \tau^j\wedge \tau^i}^{t \wedge \tau^i} 
( \sigma_i^j (X^i_s) dW^i 
+ \mu_{i}^j (X_s) \, dt) 
\end{split}
\end{equation}
for $ i= 1, 2, \cdots, n $, where
$ W^i, i=1,\cdots, n $ are independent Brownian motions, 
and for 
$ i,j =1,\cdots, n $, 
$ C_{i,j} $ are non-negative constants, 
$ \sigma_i $, $ \mu_i $, $ \sigma_{i}^j $, 
$ \mu_i^j $, each defined on $ \mathbf{R}^n $,
are smooth function 
with at most linear growth.

In a more concise way of saying, 
\begin{itemize}
\item each component 
is a diffusion process, independent to 
each other, for
each interval 
from a default time to next one,
\item the default of $ i $-th company brings about a jump $ C_{ij} $ to 
$ j $-th company,
\item \underline{which may causes the default of $ j $-th company}. 
\item The $ i $-th default may also 
affects the dynamics of the $ j $-th 
firm value process in terms of 
its growth rate or the volatility. 
\end{itemize}

Define the first contagion time by 
\begin{equation*}
\tau (1) := 
\min \{ \tau_i : i=1, \cdots, n \},
\end{equation*}
with the convention that, 
$ \min \emptyset = \infty $.
The $ j $-th contagion time is defined  recursively by
\begin{equation*}
\tau (j) := \min \{ \tau_k : \tau_k > \tau(j-1) \}, \quad j=2, 3, \cdots, n^*,
\end{equation*}
where 
$ n^* $ is the random variable 
so defined that at  
$ n^* $-the contagion time
all the companies left default. 
Note that $ n^* \leq n $ but 
each of $ \tau (j) $ can be infinity. 

To price credit derivatives
such as CDO or CDS, 
the distribution of the number of defaulted companies by a fixed time, denoted by $N_{t}$,
and the joint distribution of 
$ \tau_i, i \in I_0 \subset 
\{1,\cdots, n \} $
are required. 

These are in principle obtained from the joint distribution of
\begin{equation*}
    (\tau(1), 
    \cdots, \tau (n), d (\tau(1)), 
    \cdots, d (\tau(n) )),
\end{equation*}
where  
\begin{equation*}
    d (\tau (k))
    = \{ i \in \{1,\cdots, n \} : \tau_i = \tau (k) \}, 
    k=1, \cdots, n^*.
\end{equation*}

\section{General Case}

\subsection{Key Ideas}
The first key idea is that
we regard 
$ (\tau(i), X_{\tau (i)} ) $ as (something like) a ``renewal-reward" process. 
We shall have a formula
of the joint density 
of 
$$ ( d (\tau (1)), \tau(1), X_{\tau(1)}) $$ 
conditioned by the starting point $ X_0 $.
Here we understand 
$ X_{\tau (1) +t }, t \geq 0 $
to be an $ R^{d(\tau(1))^c} $-valued
process; we are only interested in
the survived companies.
Then, by replacing 
$ \{1,\cdots, n \} $ with 
$ d(\tau(1))^c $ and 
$ X_0 $ with $ X_{\tau(1)} $, 
we obtain the joint distribution
of $$ ( d (\tau (2)), \tau(2), X_{\tau(2)}) $$
conditioned by 
$ X_{\tau(1)} $, 
thanks to Markov property of $ X $.
We can repeat this procedure to get 
the desired joint distribution.

We can separate the problem 
of determining 
the joint distribution of 
$ ( d (\tau (1)), \tau(1), X_{\tau(1)}) $
into three parts.
\begin{itemize}
\item We pretend that 
we are given 
the harmonic measure of 
$ X_{\tau (1)-} $ (before
the ``artificial" jumps):
in a simple Brownian case
it is known. 
\item Then the problem reduces to the description
of ``contagion domain'',
but it may not be in the form of 
disjoint union.
\item To get a computable 
form, we rely on 
the independence and  
a recursive equation.
\end{itemize}

To take into account 
that we work on a ``renewal" setting 
described as above, 
from now on we let the index set of 
$ X $ be arbitrary finite subset.
In order to specify the initial 
index set, we put superscript $ I $
to the previously defined notations;
$ \tau^I (1) $, $ d^{I} $, and so on. 
We then concentrate on the study 
of the joint distribution 
of
\begin{equation}\label{www0}
(d^I (\tau^I (1)),  \tau^I (1), 
X^{I \setminus J}_{\tau^I (1)} ).
\end{equation}

\subsection{Contagion domain}

Let $ I := \{ i_1, \cdots, i_{\sharp I} \} $
and for a permutation $ \sigma \in \mathfrak{S}_{I}
$ over $ I $, 
we put
\begin{equation*}
\begin{split}
& D_{I,\sigma} := \\
& \Big\{ (x_{i_1}, \cdots, x_{i_{\sharp I}}) \in \mathbf{R}^I :  
x_{i_{\sigma(1)}} 
= K^{i_{\sigma (1)}}, x_{i_{\sigma(2)}} \in (K^{i_{\sigma(2)}}, K^{i_{\sigma(2)}} + C_{i_{\sigma(1)}, 
i_{\sigma(2)}}], \\
& \hspace{2cm}
\cdots, x_{i_{\sigma(\sharp I)}} \in (K^{i_{\sigma(\sharp I)}},  K^{i_{\sigma(\sharp I)}} + \sum_{j=1}^{ \sharp I -1} 
C_{i_{\sigma (j)}, i_{\sigma (\sharp I)}}]
\Big\}.
\end{split}
\end{equation*}
Then, we have the following 
\begin{lemma}\label{cond}
For $ \emptyset \ne J \subset I $, 
we have that
\begin{equation*}
\begin{split}
& \{ d^I (\tau^I(1)) =J \} \\
& = \left \{ X^I_{\tau^I (1)-} 
= (X^{J}_{\tau^I(1)-},
X^{I \setminus J}_{\tau^I(1)-}) \in 
\bigcup_{\sigma \in \mathfrak{S}_J} D_{J,\sigma} \times \prod_{i \in I \setminus J}
(K^i + \sum_{j \in J}
C_{j,i}, \infty) 
\right\}. 
\end{split}
\end{equation*}
\end{lemma}

\begin{proof}
The relation is clear if one sees 
that 
$ d^I (\tau^I (1))= J 
:= \{ j_1, \cdots, j_m \} $
is equivalent to the following:
there is a permutation 
$ \sigma \in \mathfrak{S}_J $
such that 

\begin{enumerate}[(i)]
\item 
$ X^j_{\sigma(1)} $ hits 
$ K^{j_{\sigma(1)}} $, 
\item at the hitting time 
$ X^{j_{\sigma(2)}} $
is in the interval 
$ (K^{j_{\sigma(2)}},K^{j_{\sigma(2)}} + C_{j_{\sigma (1)},
j_{\sigma(2)}} ] $
so that 
the $ j_{\sigma(1)} $-th company's 
default caused 
$ j_{\sigma (2)} $-th company,
\item $ X^{j_{\sigma(3)}} $
is in the interval 
$ (K^{j_{\sigma(3)}},K^{j_{\sigma(2)}} + C_{j_{\sigma (1)},
j_{\sigma(3)}} + C_{j_{\sigma (2)},
j_{\sigma(3)}}] $ 
so that 
$ j_{\sigma(3)} $-th company 
defaulted due to $ j_{\sigma(1)}
$ and/or $ j_{\sigma(2)} $-th company's default,

\qquad $ \vdots $

\item[(m)] $ X^{j_{\sigma(m)}} 
\in (K^{j_{\sigma(m)}},  K^{j_{\sigma(m)}} + \sum_{l=1}^{ m-1} 
C_{j_{\sigma (l)}, i_{\sigma (m)}}]$,
\end{enumerate}
and 
(m+1) for 
$ i \in I \setminus J $, 
$ X^{i} \in (K^i + \sum_{j \in J} C_{j,i}, \infty) $ so that 
the default of the companies 
indexed by $ J $ did not 
cause the default of 
$ i $-th company. 
\end{proof}

We put
\begin{equation*}
\begin{split}
D^I_I := \bigcup_{\sigma \in \mathfrak{S}_I} D_{I,\sigma},
\end{split}
\end{equation*}
and for non-empty $ J \subsetneq I $,
we put
\begin{equation*}
\begin{split}
D^I_J :=D^J_J \times A^I_J,
\end{split}
\end{equation*}
where 
\begin{equation}\label{AIJ}
\begin{split}
A^I_J 
:= \prod_{i \in I \setminus J}
(K^i + \sum_{j \in J}
C_{j,i}, \infty).
\end{split}
\end{equation}

\begin{lemma}\label{disj}
$ D^I_{J} \cap D^I_{J'} =
\emptyset $ if 
$ J \ne J' $.
\end{lemma}
\begin{proof}
We may assume without 
loss of generality, 
$ J \setminus J' $ is non-empty. 
Let 
$ J\setminus J'
= \{ k_1, \cdots, k_l \} $
and $ J \cap J' 
= \{ k_{l+1}, \cdots, k_{\sharp J} \} $. 
Then, for 
$ x \in D^I_J $,
there is a permutation
$ \sigma \in \mathfrak{S}_J $
such that 
\begin{equation}\label{DJdash}
\begin{split}
&(x_{k_1}, 
\cdots, x_{k_l}) \\
& \in 
(K^{k_1},
K^{k_1} + \sum_{j: \sigma^{-1}(j) < \sigma^{-1} (1)} C_{k_j,k_1}]
\times \cdots 
\times (K^{k_l},
K^{k_l} + \sum_{j: \sigma^{-1}(j) < \sigma^{-1} (l)} C_{k_j,k_l}],
\end{split}
\end{equation}
where if $ \sigma (i) =1 $,
the sum is set to be zero. 
On the other hand, 
for $ x = (x_{i_1}, \cdots, x_{i_{\sharp I}}) \in 
D^I_{J'} $, 
it holds that 
\begin{equation}\label{DJ}
(x_{k_1}, 
\cdots, x_{k_l}) \in 
(K^{k_1} + \sum_{j \in J'}
C_{j,k_1}, \infty)
\times \cdots \times 
(K^{k_l} + \sum_{j \in J'}
C_{j,k_l}, \infty)
\end{equation}
since $ \{k_1,\cdots, k_l\}  
\subset I \setminus 
J' $. 

Let 
\begin{equation*}
j_* := 
\mathop{\mathrm{argmin}}_{j \in J} \sigma^{-1} (j).
\end{equation*}
Then 
\begin{equation*}
\sigma^{-1} (j) < \sigma^{-1}(j_*) 
\end{equation*}
implies $ j \in J' \cap J $
and therefore
\begin{equation*}
\{j: \sigma^{-1} (j) < 
\sigma^{-1} (j_*) \} \subset J'.
\end{equation*}
Now we see that 
\eqref{DJdash} and 
\eqref{DJ} are not compatible, 
meaning that 
$ D^I_J \cap D^I_{J'} $
is empty. 
\end{proof}

Further, we set
\begin{equation*}
\begin{split}
D^I := \bigcup_{i \in I} 
\{K^i\} \times 
\prod_{j \ne i} (K^j,\infty).
\end{split}
\end{equation*}
Then we have the following 
\begin{lemma}\label{recf}
Let $ J_1, \cdots, J_k  $  
be such that 
$ \emptyset \ne J_k \subsetneq 
\cdots \subsetneq J_1 $,
and set 
\begin{equation*}
H_{J_1, \cdots, J_k} := 
D^{J_k} \times A^{J_{k-1}}_{J_k}
\times \cdots \times A^{J_1}_{J_2}.
\end{equation*}
Then it holds that
\begin{equation}\label{rec3}
D^I_I = H_I \setminus 
\uplus_{\emptyset \ne J_1 \subsetneq I} (H_{I,J_1} 
\setminus \uplus_{\emptyset \ne J_2 \subsetneq J_1}
( H_{I,J_1, J_2} \setminus \uplus_{\emptyset \ne J_3 
\subsetneq J_2} (H_{I, J_1, J_2, J_3}
\setminus \cdots))) 
\end{equation}
so that for any measure $ \mu $ on 
$ \mathfrak{B} (D^I) $, 
\begin{equation}\label{rec4}
\begin{split}
\mu (D^I_I) &= 
\mu (H_I) +\sum_{\emptyset \ne J_1 \subsetneq I}
(-1) \mu (H_{I,J_1}) 
+ \sum_{J_1 \subsetneq I}\sum_{\emptyset \ne  J_2 \subsetneq J_1}(-1)^2 \mu ( H_{I,J_1, J_2})
\\
& + \cdots +
\sum_{J_1 \subsetneq I} 
\cdots \sum_{\emptyset \ne J_{\sharp I-1} 
\subsetneq J_{\sharp I-2}} 
(-1)^{\sharp I-1} \mu (H_{I,J_1, \cdots, 
J_{\sharp I -1}})   . 
\end{split}
\end{equation}
\end{lemma}
\begin{proof}
We first prove 
\begin{equation}\label{rec2}
D^I = \cup_{J \subseteq I} D^I_J.
\end{equation}
That $ D^I $ includes 
the right-hand-side is clear.　
Suppose that $ x \in D^I $. 
Then, $ x^i = K^i $ for a unique $ i $ 
and $ x^j > K^j $ for $ j \ne i $. 
If there exits $ j \ne i $ such that
$ x^j \leq K^j + C_{i,j} $, we
rename it as $ j_1 $. Otherwise, 
$ (x_j)_{j \in I \setminus \{i\}} \in 
A^I_{\{i\}} $ and so $ x \in D^I_{\{i\}} $.
Among $ I \setminus\{i,j_1\} $, 
if we could find $ j $ such that 
$ x_j \leq K^j + C_{i,j} + C_{j_1,j} $, 
we rename it $ j_2 $. Otherwise 
$ x \in D^I_{\{i,j_1\}} $.
This procedure can be repeated at most 
$ \sharp I - 1 $ times when we have 
$ x \in D^I_I $. So in any case 
$ x \in \cup_{J \subseteq I} D^I_J $.   

Next, observe that
\begin{equation*}
\begin{split}
D^I_I &= D^I \setminus \uplus_{J_1 \subsetneq 
I } D^I_{J_1} \\
&= D^I \setminus \uplus_{J_1 \subsetneq 
I } (D^{J_1}_{J_1} \times A^I_{J_1})
\end{split}
\end{equation*}
by \eqref{rec2} and Lemma \ref{disj}.
By applying \eqref{rec2} to
$ D^{J_1}_{J_1} $,
and by Lemma \ref{disj} we obtain that
\begin{equation*}
D^I_{I} = 
D^I \setminus 
\uplus_{J_1 \subsetneq 
I } \left( (D^{J_1} \times A^{I}_{J_1})
\setminus 
\uplus_{J_2 \subsetneq 
J_1 } (D^{J_2}_{J_2} \times
A^{J_1}_{J_2} ) \times A^{I}_{J_1} \right).
\end{equation*}
By applying \eqref{rec2} to
$ D^{J_2}_{J_2} $ and 
Lemma \ref{disj}, and so on, 
we finally reach \eqref{rec3}.
\end{proof}

\subsection{The first main result}
Let  
$ \emptyset \ne J \subsetneq I $, and define a family of measures  
\begin{equation*}
h^I_J(x, A, S) :=
P ( d^I (\tau^I (1)) = J, 
X^{I \setminus J}_{\tau^I (1)} \in A, \tau^I (1) \in S | X^I_0 =x ), 
\end{equation*}
\begin{equation*}
h^I (x,S) = P ( d^I(\tau^I (1)) = I, 
\tau(1) \in S  | X^I_0 =x ), 
\end{equation*}
and
\begin{equation*}
Q^I (x,A,S) = P (X^I_{\tau^I (1)-} \in A, 
\tau^I (1) \in S | X^I_0 =x )
\end{equation*}
for $ x \in \prod_{i \in I} (K^i, \infty) $, $ A \in \mathfrak{B} (D^I) $ and 
$ S \in \mathfrak{B} [0,\infty) $.
The last one is the harmonic measure of the process $ X^I $
on the boundary $ \partial \prod_{i \in I} (K^i, \infty) = D^I $.

Our first main result 
relates the joint distribution 
of $ (d^I (\tau^I (1)),  \tau^I (1), 
X^{I \setminus J}_{\tau^I (1)} )
$ to the harmonic measure $ Q $.
\begin{theorem}\label{firstmain}
We have that
\begin{equation}\label{WWW-N-0}
h^I_J (x, A, S) = Q^I (x, D^J_J \times s^I_J (A), S)  
\end{equation}
and 
\begin{equation}\label{WWW-N-1}
h^{I} (x,S) = Q^I (x, D^I_I, S). 
\end{equation}
Here $ s^{I}_J $ is a shift 
on $ \mathbb{R}^{I \setminus J} $ defined by 
\begin{equation}\label{shiftIJ}
s^I_J ((x^i, i \in I \setminus J))= ((x_i + \sum_{j \in J} C_{j,i}
)).
\end{equation}
\if1
and thus 
\begin{equation*}
s^I_J ( A )
= ( A + \sum_{j \in J } 
C_{j, \cdot })
\end{equation*}
so that 
for $ a^i< b^i $, $ i \in I \setminus J $, 
\begin{equation*}
s^I_J (\prod_{j \in I \setminus J}
(a^i,b^i)) = \prod_{j \in I \setminus J} 
( \max(K^i,a^i)  + \sum_{j \in J}
C_{j,i}, b^i + \sum_{j \in J}
C_{j,i}).
\end{equation*}
\fi
\end{theorem}
\begin{proof}
The formula 
\eqref{WWW-N-0} 
is almost clear from Lemma \ref{cond} since we have
\begin{equation*}
\begin{split}
&\{d^I (\tau^I (1)) = J, 
X^{I \setminus J}_{\tau^I (1)} \in A, 
\tau^I (1) \in S\} \\
&=  \left \{ X^I_{\tau^I (1)-} \in 
\bigcup_{\sigma \in \mathfrak{S}_J} D_{J,\sigma} \times 
s^I_J (A) 
\right\} 
\cap \{\tau^I (1) \in S\}
\end{split}
\end{equation*}
for 
$ S \in \mathfrak{B} (0,\infty) $.
The formula \eqref{WWW-N-1}
is also clear from 
Lemma \ref{cond}. 
\end{proof}

\section{Joint Distribution 
of Contagions and Formulas for Contagion Probabilities}\label{CDFdefault}

In this section, we are interested in 
the distribution of the number of defaulted companies by a fixed time, denoted by $N_{t}$,
and the distribution of 
the default time of some specific firm. 
These distributions are useful for pricing of credit derivatives like CDO or CDS. 

Denote 
\begin{align*}
\mathbb{I}_j := 
I \setminus \biguplus_{l=1}^{j} d (\tau^{I} (l)) = \mathbb{I}_{j-1} \setminus d (\tau^I (j)) 
\end{align*}
the random set of indices of the firms that survived after the 
$ j $-th contagion time $\tau^{I}(j)$, for $ j=1,2,\cdots, \sharp I $. Let $U_1=\tau^{I}(1), U_2=\tau^{I}(2)-\tau^{I}(1),..., U_{k}=\tau^{I}(k)-\tau^{I}(k-1) ...$ be inter-arrival times between consecutive defaults.
We let $ U_i = \infty $ if $\tau (i) = \infty $. 

Thanks to the Markov property of the firm value evolution 
we get the following $ (X^{\mathbb{I}_j}, \tau (j), \mathbb{I}_j) $.
\begin{theorem}
Let $ J_j \subset \mathbb{I}_{j-1} $
be a non-empty set, 
$ A_j \in \mathfrak{B} (\prod_{i \in \mathbb{I}_j} [K^i, \infty)) $, and 
$ S_j \in \mathfrak{B} ([0, \infty)) $, 
for $ j =1 ,2,\cdots, n^* 
$.
(i) We have the ``renewal Markov property":
for $ m = 1, \cdots, n^* $,
\begin{align*}
&P\left(X^{\mathbb{I}_{m}}_{\tau^{J}(m)}\in A_{m}, d(\tau^{I}(m))=J_{m}, \tau^I (m) \in S_{m}\mid 
\{X^{\mathbb{I}_{j}}_{\tau^{I}(j)}, \mathbb
{I}_{j}, \tau^{I}(j); j<m \}
\right)\\
=&P\left(X^{\mathbb{I}_{m}}_{\tau^{I}(m)}\in A_{m}, d(\tau^{I}(m))=J_{m}, 
\tau^I (m) \in S_{m}\mid X^{\mathbb{I}_{m-1}}_{\tau^{I}(m-1)}, \mathbb{I}_{m-1}, \tau^{I}(m-1)
\right),
\end{align*}
and (ii) the transition probability 
is described by the harmonic measure as
\begin{align*}
\begin{aligned}
& P\left( X^{\mathbb{I}_{m} }_{\tau^I (m)} \in A_m,
d(\tau(m)) = J_m, 
\tau^I (m) \in S_m| X^{\mathbb{I}_{m-1}}_{\tau^{I}(m-1)},
\mathbb{I}_{m-1}, \tau^{I} (m-1) \right) \\
& = h^{\mathbb{I}_{m-1}}_{J_m} (X^{\mathbb{I}_{m-1}}_{\tau^I (m-1)}, A_m, S_m -\tau^I (m-1) ) \\
&= Q^{\mathbb{I}_{m-1}} (X^{\mathbb{I}_{m-1}}_{\tau^{I}(m-1)},
D^{J_m}_{J_m} \times 
s^{\mathbb{I}_{m-1}}_{J_m}(A_m),S_m- \tau^I (m-1) ) ). 
\end{aligned}
\end{align*}
(iii) Consequently, we have that 
\begin{align*}
&P\left( \{ X^{\mathbb{I}_{j}}_{\tau^{I}(j)}\in A_j, 
d(\tau^{I}(j))=J_{j}, U_{j}\in S_{j};  j\leq m \} \mid X^{I}_{0}=x^{I}
\right)\\
&= \int_{A_1 \times \cdots \times A_m} \prod_{j=1}^{m} h^{{I}_{j-1}}_{J_j} \left(
x^{{I}_{j-1}}_{j-1},  
dx^{I_j}_j, S_j
\right) \\
&=
\int_{\prod_{j=1}^m 
s^{I_{j-1}}_{J_j} (A_j)
} \prod_{j=1}^{m}Q^{{I}_{j-1}}\left(
(s^{I_{j-2}}_{J_{j-1}})^{-1}
(x^{I_{j-2}\setminus J_{j-1}
= {I}_{j-1}}_{j-1}), 
D^{J_{j}}_{J_{j}}\times 
dx^{I_{j-1}\setminus J_j=I_j}_j, S_j
\right),
\end{align*}
where
\begin{align*}
I_0 = I, \quad I_j := 
I \setminus \biguplus_{l=1}^{j} J_l, 
\quad j=1,2, \cdots, m \leq n^*,
\end{align*}
and $ s^{I_*}_{J_{**}} $
is the shift defined 
in the previous section as \eqref{shiftIJ},
with the convention
that $ s^{I_{-1}}_* $ is the identity map. 
\if2 
\begin{align*}
A'_j := A_j + \left(\sum_{p\in J_{j}}C_{p,.}\right)^{\mathbb{I}_{j}},
j=1,\cdots, \sharp I,
\end{align*}
\fi
\end{theorem}
\begin{proof}
The first assertion (i)
is clear from the Markov property of $ X $. 
The second one 
is also a direct consequence of 
the time-homogeneous property and Theorem \ref{firstmain}.
The third one is obtained by
combining 
(ii) and the renewal Markov property (i).
\end{proof}

As a consequence, we can get the 
``marginal distribution" concerning on inter-arrival time between two consecutive defaults and the set of the next default index.
\begin{corollary}\label{JointInter}
For $ J_1, \cdots, J_m \subset I $ with $ \emptyset \ne J_m^c \subsetneq \cdots \subsetneq J_1^c $, and 
$ S_1, \cdots S_m \in \mathfrak{B} ([0, \infty)) $, 
\begin{align*}
&P( \{ U_{j}\in S_{j}, d(\tau^{I}(j))=J_{j}; j\leq m \})\\
=&\prod_{j=1}^{m}\int_{A^{I_{j-1}}_{J_{j}}}
Q^{{I}_{j-1}}\left(
(s^{I_{j-2}}_{J_{j-1}} )^{-1} (x_{j-1}), 
D^{J_{j}}_{J_{j}}\times
dx_{j}, S_j \right),
\end{align*}
where $ A^{I_*}_{J_{**}} $
is the set defined in the previous section as \eqref{AIJ}.
\end{corollary}

Corollary \ref{JointInter} is an important key to find out some important distributions
such as the number of defaulted
firms given a fixed time, the time to default of a given firm,
or the time to the $m$-th  default. 
The following proposition provides the distribution for the number of defaulted firms 
given a fixed time.
\if3
We write
\begin{align*}
R^{I}(x, t)=P(
\tau^{I}(1)>t|X^{I}_0=x),
\end{align*}
the survival distribution function of the first time to default. 
\fi
\begin{proposition}\label{NTK}
For $k=1,..., n$, we have
\begin{align*}
& P(N_{t}=k\mid X^{I}_{0}=x^{I}_{0})\\
&=\sum_{m=1}^{k}
\sum_{ \sharp \biguplus_{p\leq m } J_p = k }
\int_{
\substack{u_1+\cdots +u_m \leq  t \\
u_1 +\cdots +u_{m+1} >t }} 
\prod_{j=1}^{m} 
\\
& \hspace{2cm} 
\int_{A^{I \setminus \uplus_{l=1}^{j-1} J_l}_{J_j}}
Q^{{I \setminus \uplus_{l=1}^{j-1} J_l}}\left(
(s^{{I \setminus \uplus_{l=1}^{j-2} J_l}}_{J_{j-1}})^{-1} (x_{j-1}), 
D^{J_{j}}_{J_{j}}\times
dx_{j}, du_j \right). \\
\end{align*}
\end{proposition}

\begin{proof}
Since the event $\{N_t=k\}$
can be expressed as $\biguplus_{m=1}^{k}
\{\tau^{I}(m)\leq t<\tau^{I}(m+1), \# (I \setminus \mathbb{I}_{m})
= \sum_{p=1}^m \# d (\tau ^{I}(p)) =k\}$,
we have that 
\begin{align}\label{formula5-1}
\begin{aligned}
&P(N_t=k\mid X^{I}_{0}=x^{I}_{0}) \\
&=\sum_{m=1}^{k}P \Big(\tau^{I}(m)\leq t<\tau^{I}(m+1), \sum_{p=1}^m \# d (\tau^{I} (p)) =k \mid X^{I}_{0}=x^{I}_{0} \Big) \\
&=\sum_{m=1}^k
\sum_{ \sharp \biguplus_{p\leq m } J_p = k }
P(\sum_{j=1}^m U_j \leq t<
\sum_{j=1}^{m+1} U_j, 
d (\tau^{I} (p)) = J_p, 
p\leq m 
\mid X^{I}_{0}=x^{I}_{0}).
\end{aligned}
\end{align}
By applying Corollary \ref{JointInter},
we have that 
\begin{align*}
\begin{aligned}
& (\text{the summand 
of \eqref{formula5-1}})\\
&= \int_{
\substack{u_1+\cdots +u_m \leq  t \\
u_1 +\cdots +u_{m+1} >t }} 
\prod_{j=1}^{m}\int_{A^{I \setminus \uplus_{l=1}^{j-1} J_l}_{J_j}}
Q^{{I \setminus \uplus_{l=1}^{j-1} J_l}}\left(
(s^{{I \setminus \uplus_{l=1}^{j-2} J_l}}_{J_{j-1}})^{-1} (x_{j-1}), 
D^{J_{j}}_{J_{j}}\times
dx_{j}, du_j \right).
\end{aligned}
\end{align*}
\end{proof}

We can also obtain the distribution function of the $m$-th contagion time.
\begin{proposition}
For $m=1,..., n$, we have
\begin{align}\label{kth_time}
\begin{aligned}
& P(\tau(m))>t\mid X^{I}_{0}=x^{I}_{0}) \\
&= \sum_{k=1}^{m-1}\sum_{\#\uplus_{i=1}^{k-1} I_{i}<n}
\int_{
\substack{u_1+\cdots +u_m \leq  t \\
u_1 +\cdots +u_{m+1} >t }} \\
& \hspace{1.5cm}
\prod_{j=1}^{m}\int_{A^{I \setminus \uplus_{l=1}^{j-1} J_l}_{J_j}}
Q^{{I \setminus \uplus_{l=1}^{j-1} J_l}}\left(
(s^{{I \setminus \uplus_{l=1}^{j-2} J_l}}_{J_{j-1}})^{-1} (x_{j-1}), 
D^{J_{j}}_{J_{j}}\times
dx_{j}, du_j \right). \\
\end{aligned}
\end{align}
\end{proposition}
\begin{proof}
We have 
\begin{align*}
&P(\tau^{I}(k)>t\mid X^{I}_{0}=x^{I}_{0})\\
=&\sum_{m=1}^{k}P(\tau^{I}(m-1)\leq t<\tau^{I}(m)\mid X^{I}_{0}=x^{I}_{0})\\
=&\sum_{m=1}^{k}\sum_{\#\biguplus J_{p}<n}P(\tau^{I}(m-1)\leq t<\tau^{I}(m), d(\tau^{I}(p))=J_{p}, p\leq m\mid X^{I}_{0}=x^{I}_{0}).
\end{align*}
Using the same technique
as we did for 
Proposition \ref{NTK}, we obtain the formula \ref{kth_time}.
\end{proof}

The $k$-th company survives
up to time $t$ if and only if 
the index $k$ is always out of 
all the default set $d(\tau^{I}(p))$ up to time $t$. 
Hence we can get 
survival probability 
of a set of firms as
\begin{proposition}
Let $ K \subset I $ be a non-empty set. Then we have that 
\begin{align*}
\begin{aligned}
&P( \tau_{k}>t, k \in K \mid X^{I}_{0}=x^{I}_{0}) \\
&= \sum_{m=0}^{n- \sharp K}
\sum_{\uplus_{l=1}^p J_l \cap K = \emptyset }
\int_{
\substack{u_1+\cdots +u_m \leq  t \\
u_1 +\cdots +u_{m+1} >t }} \\
& \hspace{2cm}
\prod_{j=1}^{m}\int_{A^{I \setminus \uplus_{l=1}^{j-1} J_l}_{J_j}}
Q^{{I \setminus \uplus_{l=1}^{j-1} J_l}}\left(
(s^{{I \setminus \uplus_{l=1}^{j-2} J_l}}_{J_{j-1}})^{-1} (x_{j-1}), 
D^{J_{j}}_{J_{j}}\times
dx_{j}, du_j \right).
\end{aligned}
\end{align*}
\end{proposition}

\begin{proof}
The formula can be obtained 
in a similar way as the previous propositions 
by noting 
\begin{align*}
\begin{aligned}
&P( \tau_{k}>t, k \in K \mid X^{I}_{0}=x^{I}_{0}) \\
&= \sum_{m=0}^{n- \sharp K}
\sum_{\uplus_{l=1}^p J_l \cap K = \emptyset }
P(\tau^{I}(m)\leq t<\tau^{I}(m+1), d(\tau^{I}(p))=J_{p}, p\leq m\mid X^{I}_{0}=x^{I}_{0}),
\end{aligned}
\end{align*}
where we understand $ J_0 =\emptyset $ when $ m=0 $.
\end{proof}

\section{Models with Independence}

\subsection{The ``harmonic measure"}

As we have seen,  
the joint distributions we need is 
obtained from 
the ``harmonic measure" $ Q^I $.
In this section 
we impose independence among $ X^i $.
Then the 
it is in fact expressed in terms of 
the harmonic measures of $ X^i $ to
$ [K^i, \infty) $, $ i \in I $.

Let us be more precise. 
Let $ \tilde{X} $
be a kind of {\em Business As Usual}
process given as 
\begin{equation*}
\begin{split}
\tilde{X}^{i}_t &= x^i 
+ \int_0^t 
( \sigma_i (\tilde{X}^i_s) dW^i
+ \mu^i ( \tilde{X}_s^i) \, dt),
\end{split}
\end{equation*}
and $ \tilde{\tau}_i $
be its default time:
\begin{equation*}
\tilde{\tau}_i
:= \inf\{
s>0: \tilde{X}^{i}_s \leq K^i \}.
\end{equation*}
We assume that each of 
the distribution of 
$ (\tilde{\tau}_i, \tilde{X}^i ) $
has a density, and put
\begin{equation*}
p_j (x^i,s) = 
\frac{P ( \tilde{\tau}_i \in ds
| \tilde{X}^i= x^i)}{ds}, 
\end{equation*}
and 
\begin{equation*}
q_i (x^i,y_i,s) 
= \frac{P(\tilde{\tau}_i > s, \tilde{X}^i_s \in dy_i| \tilde{X}^i= x^i)}{dy_i},
\end{equation*}
for $ x^i, y^i \in [K^i, \infty) $
and $ s >0 $.

The ``harmonic measure'', 
the distribution of 
of $ X^I_{\tau(1)} $
can be obtained by the following 
\begin{lemma}\label{hm00}
For $ A \in \mathfrak{B} (G) $
and $ S \in \mathfrak{B} (0,\infty) $, 
\begin{equation}\label{hm01}
\begin{split}
& Q (x,A, S) 
=
\int_S \sum_i  
p_i(x^i,s) ds \int_A \delta_{K^i} (dy_i) 
\prod_{j \ne i}
q_j (x^j, y_j, s) dy_j, \\
\end{split}
\end{equation}
where $ \delta_{\ast} $ is the Dirac delta
at $ \ast $.
\end{lemma}

\begin{proof}
The left-hand-side of \eqref{hm01}
\begin{equation*}
\begin{split}
&  = \int_S \sum_i P ( \{\tilde{\tau}_i\in ds\}
\cap_{j \ne i} \{ \tilde{\tau}^i < \tilde{\tau}^j, \tilde{X}_{\tilde{\tau}_j} \in A \} ) \\
& = \int_S \sum_i \int_{A} \delta_{K^i} (dy_i)
P ( \tilde{\tau}^{i} \in ds, s< \tau^{j}, \tilde{X}_s^{j} \in dy_j,
\forall j \ne i). \\
\end{split}
\end{equation*}
By the independence of $ \tilde{X}^{i}, 
i \in I $, 
we have the desired relation \eqref{hm01}. 
\end{proof}

\ 

We put 
\begin{equation*}
g_{J}^I (x^{I \setminus J}, A,s)  :=
\int_{\prod_{i \in I \setminus J}
[K^i, \infty) \cap A }
\prod_{i \in I \setminus J} q_i (x^i,
y_i
+ \sum_{j \in J}
C_{j,i},s) \, dy_i,
\end{equation*}
and
\begin{equation*}
\begin{split}
g^{J}_I (x^{I \setminus J},s)&:= 
g^{I}_J (x^{I \setminus J}, 
\mathbf{R}^{I \setminus J},s ) \\
&= \int_{A^I_J} \prod_{i \in I \setminus J} q_i (x^i,
y_i,s) \,dy_i
\end{split}
\end{equation*}
for $ s> 0 $ and $ A \in 
\mathfrak{B} (\mathbf{R}^{I \setminus J})$.

Since here $ h^I $ 
have a density, we will write, 
with a slight abuse of the notations, 
\begin{equation*}
h^I (x^I,s) = h^I (x^I,ds)/ds \equiv P ( d^I (\tau^I (1)) = I, \tau^I (1) \in ds | X^I_0 = x^I )/ds, \quad s>0.
\end{equation*}

\begin{theorem}\label{hik}
(i) 
For a 
non-empty 
$ J \subsetneq I $,
$ S \in  \mathfrak{B} (0,\infty) $,  
and $ A \in 
\mathfrak{B} (\mathbf{R}^{I \setminus J}) $,
\begin{equation}\label{HI}
h^I_J (x=x^I,A,S) = \int_S
h^J (x^J, s) g^{I}_J (s,x^{I \setminus J},A) ds. 
\end{equation}
(ii) For $ s > 0 $, 
\begin{equation}\label{www}
\begin{split}
        & h^I (x=x^I,s) \\
        &= 
        \left( \sum_{i \in I }p_i (x^i, s) \right)
        \Big( 1+\sum_{m=1}^{\sharp I-1}(-1)^{m}\sum_{\substack{I_{m}\subsetneq ... \subsetneq I_{1} \subsetneq I_0 :=I }}
         \prod_{l=1}^m g_{I_l}^{I_{l-1} }(x^{I_{l-1} \setminus I_l }, s)
         \Big).
         \end{split}
\end{equation}
\end{theorem}

\begin{proof}
(i) It suffices to show 
when $ A = \prod_{i \in I \setminus J} (a_i, b_i) $. 
By combining \eqref{WWW-N-0}
and Lemma \ref{hm00}, 
we see that
\begin{equation*}
\begin{split}
& h^I_J (S, \prod_{i \in I \setminus J} (a_i, b_i) ) \\
&= \int_{S} \sum_{i \in J}  
p_i(s) \int_{D^J_J
} \delta_{K^i} (dx_i) 
\prod_{j \in J \setminus \{i\} }
q_j (s, x_j) \,dx_j \\
& \qquad \times \prod_{i \in I \setminus J}
\int_{(a_i
+  \sum_{j \in J} 
C_{j,i}, b_i + \sum_{j \in J} 
C_{j,i}  ) \cap  
(K^i + \sum_{j \in J} 
C_{j,i}, \infty) }
q_i (s, x_i) dx_i \\
&= \int_S h^J (s) 
\int_{ \prod_{i \in I \setminus J}[K^i, \infty) \cap (a_i, b_i) }
\prod_{i \in I \setminus J} q_i (s, x_i -\sum_{j \in J} 
C_{j,i} ) dx ds. 
\end{split}
\end{equation*}
Thus we obtained \eqref{HI}.

(ii) The equation \eqref{www}
is a direct consequence 
of \eqref{WWW-N-1} in Theorem \ref{firstmain}
and \eqref{rec4} in Lemma \ref{recf},
together with the relation 
\begin{equation*}
g^{I}_J (s,x=x^{I \setminus J})
= \int_{A^I_J} \prod_{i \in I \setminus J} q_i (s,x^i,
y_i) \,dy_i.
\end{equation*}
\end{proof}

\bibliography{Bib_credit_risk}

\end{document}